\documentclass[conference]{IEEEtran}
\IEEEoverridecommandlockouts 

\usepackage{verbatim} 
\usepackage{hyperref} 
\usepackage{graphicx, caption}
\usepackage{amssymb}
\usepackage{amsmath}
\usepackage{amsthm}
\usepackage{mathtools}
\usepackage{cite}
\usepackage{stfloats}
\usepackage{subfigure}
\usepackage{epstopdf}
\usepackage{psfrag}
\usepackage[mathscr]{euscript}
\usepackage{acronym}  
\usepackage{booktabs} 
\usepackage[table]{xcolor}

\usepackage{color}
\usepackage{dsfont}
\usepackage{bbm}

\acrodef{AoI}{Age of information}
\acrodef{PAoI}{peak age of information}
\acrodef{BS}{base station}
\acrodef{HLF}{Hyperledger Fabric}
\acrodef{HPPP} {homogeneous Poisson Point Process}
\acrodef{HeMN}{Hyperledger Fabric Blockchain-enabled monitoring network}
\acrodef{MVCC}{multi-version concurrency check}
\acrodef{STP}{successful transmission probability}
\acrodef{FCFS}{first-come-first-serve}
\acrodef{pdf}{probability density function}
\acrodef{i.i.d.}{independent and identically distributed}
\acrodef{CDF}{cumulative distribution function}

\usepackage{color}
\usepackage{dsfont}
\usepackage{bbm}

\newtheorem{theorem}{Theorem}
\newtheorem{lemma}{Lemma}

\newcommand{\E}{\mathbb{E}}
\renewcommand{\P}{\mathbb{P}}
\newcommand{\integral}[2]{\int_{#1}^{#2}}
\newcommand{\summation}[2]{\sum_{#1}^{#2}}

\allowdisplaybreaks 

\begin{document}
	\newcommand{\paperTitle}{Title}



\title{\vspace{-0.1cm} Age of Information Analysis in Hyperledger Fabric \\
	Blockchain-enabled Monitoring Networks}

\author{
	
\IEEEauthorblockN{
	Minsu~Kim, Sungho~Lee, Chanwon~Park, and 
	Jemin~Lee
}\\[-0.5em]
\vspace{-3.0mm}
\IEEEauthorblockA{
	Department of Information and Communication Engineering (ICE)\\
	Daegu Gyeongbuk Institute of Science and Technology (DGIST), Korea\\
	\vspace{-3.0mm}
	 Email: {kms0603@dgist.ac.kr}, {seuho2003@dgist.ac.kr}, {pcw0311@dgist.ac.kr},      {jmnlee@dgist.ac.kr}
}
%

}

\maketitle 

%

%

%

%

\acresetall
	
	\begin{abstract}

	\ac{AoI} is a recently proposed metric for quantifying data freshness in real-time status monitoring systems where timeliness is of importance. In this paper, we explore the data freshness in the Hyperledger Fabric Blockchain-enabled monitoring network (HeMN) by leveraging the AoI metric. In HeMN, status updates from sources are transmitted through an uplink and recorded in a Hyperledger Fabric (HLF) network. To provide a stochastic guarantee of data freshness, we derive a closed-form of AoI violation probability by considering the transmission latency and the consensus latency. Then, we validate our analytic results through the implemented HLF platform. We also investigate the effect of the target successful transmission probability (STP) on the AoI violation probability.	
	\end{abstract}

	\section{Introduction}
	Real-time status monitoring systems have been deployed worldwide for time-critical applications such as traffic management and pollution detection. These systems usually consist of a source and a destination. Specifically, the source generates a status update (or packet) with observed information and transmits it to the destination where status updates (or packets) are received and recorded. In such cases, decisions are made based on the observed information from the source, but outdated data may lead to incorrect decisions.
	Therefore, it is of importance for real-time status monitoring systems to provide fresh data, in order to prevent undesirable outputs. To quantify data freshness, \ac{AoI} has been proposed in \cite{SaRoMa:12}. It is defined as the elapsed time from the generation of the latest received status updating packet.

	The \ac{AoI} has been analyzed for real-time IoT applications in order to capture the timeliness of monitoring data \cite{MAHD:19, XSZN:19, YGBV:19}. In \cite{MAHD:19}, the optimal sampling policy is obtained to minimize the weighted sum-\ac{AoI} in energy harvesting enabled IoT networks. In \cite{XSZN:19}, the average \ac{AoI} penalty function of an energy harvesting IoT system is obtained and the status update frequency is optimized to minimize the average \ac{AoI} penalty function. An \ac{AoI}-energy trade-off is studied for IoT monitoring systems and the average \ac{AoI} is minimized by optimizing the transmission power of an IoT device in \cite{YGBV:19}. 
	However, all of the previous papers did not consider the case when information is stored in a blockchain, which cannot guarantee data integrity.
	
    Recently, a blockchain has been regarded as a promising decentralized data management platform for IoT devices aiming to eliminate the need for a central authority \cite{LeDo:19 ,PDPP:19, YSMAI:19}. A blockchain integrated IoT platform is presented for real-time monitoring to provide data integrity \cite{LeDo:19}. In \cite{PDPP:19}, the communication cost of periodic updates is analyzed in Ethereum blockchain for lightweight IoT devices, which only store the head of blocks. For the blockchain-enabled wireless IoT network, the optimal deployment of the nodes for maximizing transaction throughput is investigated in \cite{YSMAI:19}. However, data freshness in blockchain platforms is not analyzed in the previous papers \cite{LeDo:19, PDPP:19, YSMAI:19}, which fails to show the timeliness of stored data.
    
    A blockchain may be selected as the underlying system of an alarm system (e.g., pollution or fire detection) in order to ensure data integrity. However, this system requires not only the integrity of monitored data but also the timeliness of it. If any of the two properties are not satisfied, the system may result in the wrong decision. 
    Therefore, we analyze the data freshness in a \ac{HLF} platform, which is one of the most utilized permissioned blockchain platforms. In \ac{HLF}, the latest monitored information is stored in a distributed ledger with its history, which can provide data integrity.
    
    In this paper, we investigate the data freshness in the \ac{HeMN}, where sources observe physical phenomena and update their status in a \ac{HLF} network, which is connected at \acp{BS}. Sources and \ac{BS}s are distributed randomly on the network, and each source transmits a packet to its nearest \ac{BS}. To measure data freshness in \ac{HeMN}, the \ac{AoI} is utilized. We also analyze the \ac{AoI} violation probability, which shows the probability that the target \ac{AoI} is guaranteed. Our main contributions are as follows: $1)$ we derive the \ac{AoI} violation probability in a closed-form by considering the consensus latency as well as the transmission latency, $2)$ we also obtain the experiment results by implementing the \ac{HLF} platform v1.3 to verify analysis results, and $3)$ we explore the effect of the target \ac{STP} on the \ac{AoI} violation probability of \ac{HeMN}. 
    
	

	\section{\ac{HLF} Transaction Flow}
	\label{sec: transaction flow}
	
	In this section, we provide the overall structure of \ac{HLF} and the components of the consensus process for a status update. In \ac{HLF}, all changes made by transactions are committed to a distributed ledger. The ledger is a key-value database owned by peers in \ac{HLF}, and transactions can update stored data with each corresponding key. Hence, all data is identified by each own key and a version number. In \ac{HLF}, participants are all identified. Therefore, the costly consensus method in a public blockchain, which is known as mining, is not necessary for \ac{HLF}. Instead, the consensus process in \ac{HLF} is composed of three phases: endorsing phase, ordering phase, and validation phase as described below. More detailed explanations on each phase are available in \cite{SLMKJL:20} 
	\subsubsection{Endorsing Phase}
	The endorsing phase is to receive endorsements from the peers, which are entitled to simulate transactions against their own copied ledgers. The peers make sure that they have the exactly identical transaction simulation results. Then the transaction with the endorsements is transmitted to the ordering node.	
	
	\subsubsection{Ordering Phase}
	The ordering phase is not only to arrange transactions in chronological order but also to generate new blocks with the ordered transactions. The ordering nodes continuously include transactions into a new block until it reaches the pre-defined maximum block size. To avoid high latency, a timer is prepared with the pre-defined timeout value. If the timer expires, the nodes instantly export the new block, regardless of the current number of transactions in the block.
	
	\subsubsection{Validation Phase}
	The validation phase is to validate newly delivered blocks to the peers. This phase consists of two sequential steps: verification and update. Firstly, the peers investigate if each transaction is properly endorsed according to the endorsement phase. Then, the peers check whether the key versions are identical to the ones currently stored in their copied ledgers. This verification is also called the \ac{MVCC} verification. Note that the key version increases each time the corresponding status is updated. In case the version numbers are different, the transaction is marked as invalid as well as become ineffective.  Lastly, the peers update the status retained in their ledger.
	
	\begin{figure}
			\vspace{-0.cm}
			\includegraphics[width=0.93\columnwidth]{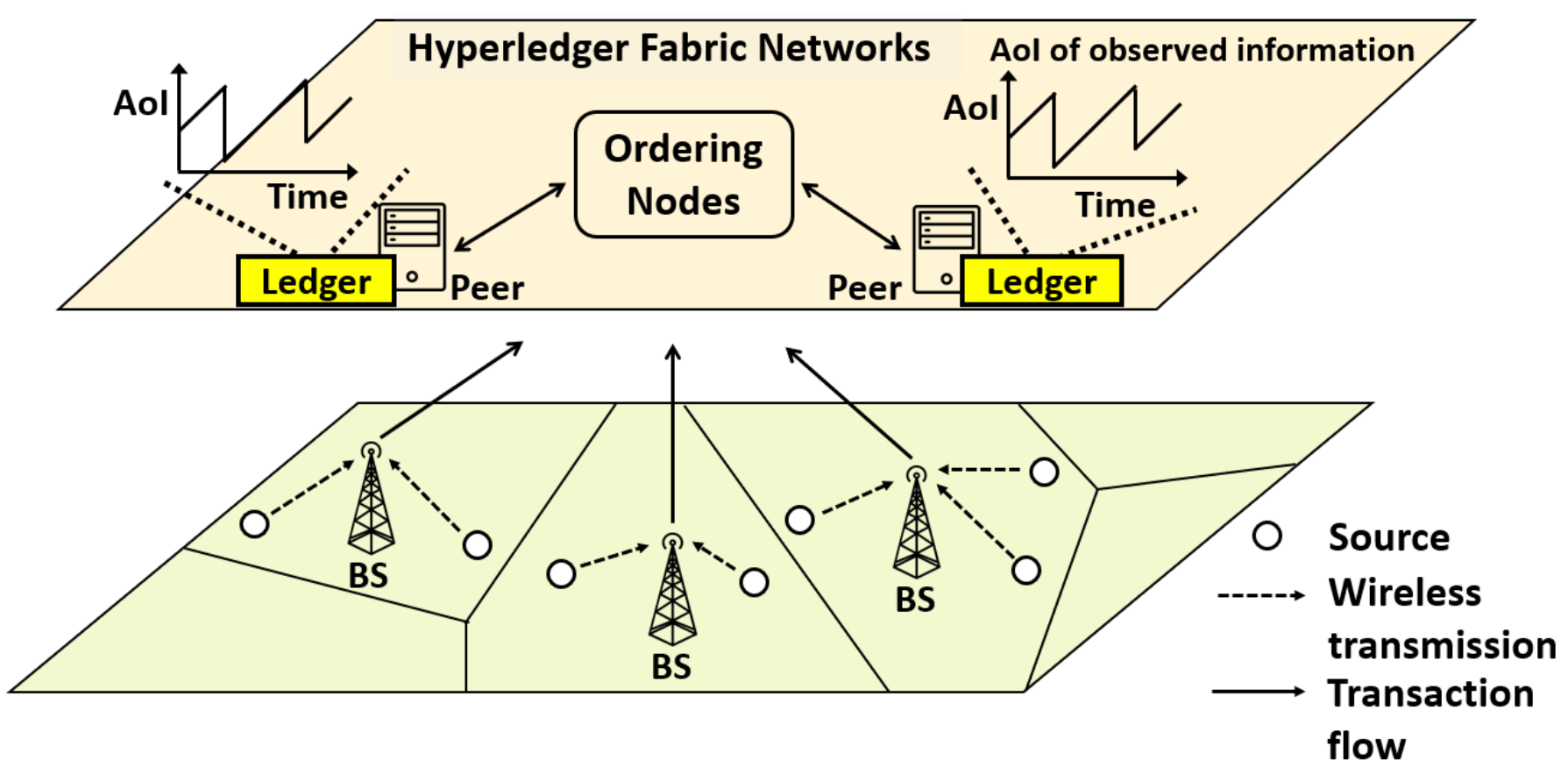}
 			\captionsetup {singlelinecheck = false}

		\caption{An illustration of \ac{HeMN}  }
		\label{fig:overall_sysmtem_model}

		\vspace{-0.4cm}
	\end{figure}

	\section{HLF Blockchain-Enabled Monitoring Network}

	\subsection{System Model}
	We consider the \ac{HeMN} composed of sources, \acp{BS}, and \ac{HLF} network, where sources monitor physical phenomena (e.g., temperature or pollution level) and update the corresponding status stored in \ac{HLF} as shown in Fig. \ref{fig:overall_sysmtem_model}. We assume the distribution of the sources follows a \ac{HPPP} $\Phi_s$ with spatial density $\lambda_s$. The source transmits a packet through a wireless uplink channel to the nearest \ac{BS}. We assume that all the sources use the same transmission power $P$. The distribution of \ac{BS}s also follows a \ac{HPPP} $\Phi_m$ with spatial density $\lambda$. Each channel is allocated to one source only in the cell of a \ac{BS} to avoid the interference between the sources in the same cell.
	%
	%
	%
	%
	
	We assume \ac{BS}s are connected to \ac{HLF} network, where the status information of sources are stored with their own key values. As shown in Fig. \ref{fig:overall_sysmtem_model}, a source monitors a physical phenomenon and generates a packet with newly observed information. The packet is delivered to \ac{HLF} via a \ac{BS} in a form of a transaction.  Successfully received transactions can update their status information through the consensus process, which is described in Sec. \ref{sec: transaction flow}.

    We define the consensus latency as the total time for the commitment of a transaction, which is the summation of the latencies in each phase in Sec. \ref{sec: transaction flow}. Then, the total latency of the $k\text{th}$ packet can be defined as 
    \begin{equation}
    Z_{k} = X_k + Y, \label{total_latency}
    \end{equation}
    where $X_k$ is the consensus latency of the $k\text{th}$ packet and $Y$ denotes the transmission latency, which is required for transmitting a packet from a source to its associated \ac{BS}.
    The consensus latency $\{X_k, k \geq1 \}$ is assumed to be \ac{i.i.d.}.
    From empirical results of a constructed \ac{HLF} platform, the consensus latency for the $k$th packet is modeled as a Gamma random variable, of which pdf is given by
    \begin{align}
    f_{X_k}(x) = \frac{\beta^\alpha}{\Gamma(\alpha)} x^{\alpha - 1} e^{-\beta x}, \label{gamma_pdf}
    \end{align}
    where $\alpha$ is the shape parameter and $\beta$ is the rate parameters.
    To determine the values of $\alpha$ and $\beta$ for the consensus latency, we use the maximum likelihood estimation \cite{Thom:58}. 
	\subsection{Analysis of Transmission Latency}
	
	In this subsection, we analyze the transmission latency $Y$ of \ac{HeMN} \cite{ChJe:18}. The signal to interference plus noise ratio (SINR) received by a \ac{BS} at $\textbf{y}_\text{o}$ from a source located at $\textbf{x}_\text{o}$ under Rayleigh fading channel is given by
	\begin{align}
	\text{SINR} = \frac{P h_{\textbf{x}_\text{o}, \textbf{y}_\text{o}}  L_{\textbf{x}_\text{o}, \textbf{y}_\text{o}}^{-n}} {I + N_0 W}, \label{SINR_definition}
	\end{align}
	where $h_{\textbf{x}_\text{o}, \textbf{y}_\text{o}}$ is the fading channel gain, i.e., $h_{\textbf{x}_\text{o}, \textbf{y}_\text{o}} \sim \exp(1)$,  $N_0$ is the noise power, $W$ is the channel bandwidth, $L_{\textbf{x}_\text{o}, \textbf{y}_\text{o}}$ is the distance between the source at $\textbf{x}_\text{o}$ and the associated \ac{BS} at $\textbf{y}_\text{o}$, and $n$ is the pathloss exponent. In \eqref{SINR_definition}, $I$ is the interference from other sources that use the same uplink frequency band, given by
	%
	%
	\begin{equation}
	I = P \summation{\textbf{x} \in \Psi_u \setminus \textbf{x}_\text{o} }{} h_{\textbf{x}, \textbf{y}_\text{o}} L^{-n}_{{\textbf{x}, \textbf{y}_\text{o}}},
	\end{equation}
	where $\Psi_u$ denotes the set of locations of the sources which use the same frequency band with the source at $\textbf{x}_\text{o}$. 
	For the data rate $R$, \ac{STP} $p_c$ is given by
	\begin{align}
	 p_c = \P \left [W \log_2(1 + \text{SINR}) \geq R \right ].  \label{STP}
	\end{align}
	To guarantee a certain level of \ac{STP}, we can set the target rate $\bar{R}$ as
	\begin{align}
	p_c = \P \left [W \log_2(1 + \text{SINR}) \geq \bar{R} \right ] \geq \zeta, \label{target_STP}
	\end{align}   
	where $\zeta$ is the target \ac{STP}.
	
	Using \eqref{target_STP}, we define the transmission latency $Y$ as $D/ \bar{R}$, where $D$ [bits] is the packet size.
	%
	%
	Then $\bar{R}$ can be obtained by
	\begin{align}
	\bar{R} = \integral{0}{\infty} \bar{R}(r) \ 2 \lambda \pi r \exp(-\lambda \pi r^2) dr, \label{Transmission_rate}
	\end{align}
	where $\bar{R}(r)$ is the target rate when the distance between the source and the \ac{BS} is fixed to r, i.e., $L_{\textbf{x}_o, \textbf{y}_o} = r$. Note that the distribution of the distance from the source and the nearest \ac{BS} follows a Rayleigh distribution. From the definition of SINR in \eqref{SINR_definition}, $p_c$ can be given by
	\begin{align}
	%
	%
	p_c = \exp\left( \hspace{-0.5mm} -\frac{r^n}{P} N_0 W \theta \hspace{-0.5mm} \right)  \E_I \left [\exp \left (\hspace{-0.5mm}  -\frac{r^n}{P} I \theta \hspace{-0.5mm}  \right) \right ],   \label{p_c_before_laplace}
	%
	\end{align}
	where $\theta = 2^{\bar{R}(r)/W}-1$ by using the \ac{CDF} of the exponential random variable $h_{\textbf{x}_o, \textbf{y}_o}$. The dependence exists among the locations of interfering sources, so their distribution does not follow \ac{HPPP}. Nevertheless, it is shown that this dependency can be weak \cite{ThHaDhJe:13}, so we assume the distribution of the interfering sources follows the \ac{HPPP} with spatial density $\lambda$. Note that a channel is allocated to one source only in the cell of a \ac{BS}, so the  density of uplink interfering sources is $\lambda$. According to \cite[eq.3.21]{Ha:09}, Laplace transform of $I$ can be given by
	\begin{align}
	\mathcal{L}_I (s) = \exp \left\{ -\lambda \pi s^{2/n} \frac{2 \pi }{n \sin(2 \pi /n)}      \right\}. \label{Laplace_interferece}
	\end{align}
	From \eqref{Laplace_interferece}, $p_c$ in \eqref{p_c_before_laplace} can be presented as
	\begin{align}
	p_c = \exp\left( \hspace{-0.5mm} -\frac{r^n}{P} N_0 W \theta \hspace{-0.5mm} \right) \exp \left\{ -\lambda \pi^2 \frac{2r^2 \theta^{2/n}}{n \sin \left(2\pi/n  \right)}   \right\}. \label{p_c_general_pathloss}
	\end{align}
    To obtain a closed-form of $\bar{R}(r)$ from \eqref{target_STP} and \eqref{p_c_general_pathloss}, we set $n=4$ for tractability. We can then present $\bar{R}(r)$ as
	\begin{align}
	\bar{R}(r) \hspace{-0.5mm}  &=  \hspace{-0.5mm}  W \hspace{-0.5mm}  \log_2 \hspace{-0.5mm}  \left[ \hspace{-0.5mm}  1  \hspace{-1.0mm}  + \hspace{-1.0mm}   \left \{ \hspace{-0.5mm}   \frac{P (-\pi^2 \lambda  \hspace{-0.5mm}  +  \hspace{-0.5mm}  \sqrt{\pi^4 \lambda^2  \hspace{-0.5mm}  -  \hspace{-0.5mm}  16N_0 W \log \zeta })} {4 N_0 W r^2} \hspace{-0.5mm}   \right\}^2   \right] \label{transmissionrate fixed r} \nonumber \\
	&= W \log_2 \left\{ 1 + \left( \frac{m}{r^2} \right)^2   \right\},
	\end{align}
	where $m$ is:
	\begin{align}
	m = \frac{P (-\pi^2 \lambda + \sqrt{\pi^4 \lambda^2 - 16N_0 W \log \zeta })} {4 N_0 W}.
	\end{align}
	By substituting \eqref{transmissionrate fixed r} with \eqref{Transmission_rate} and replacing $r^2$ with $t$, $\bar{R}$ is given by
	\begin{align}
	 &\bar{R} = \lambda  \pi W \integral{0}{\infty} \log_2 \left\{ 1 + \left( \frac{m}{t} \right)^2   \right\} \exp({-\lambda \pi t}) dt \nonumber \\
    %
	%
	&= \frac{\lambda \pi W}{\log 2} \integral{0}{\infty} \left[   \log \left\{ t^2 +   m^4   \right\} -\log (t^2) \right] \exp({-\lambda \pi t}) dt \nonumber \\
	%
	%
	&\overset{(a)}{=} \frac{W}{\log 2} \left \{\log m - \text{Ci}(m \lambda \pi) \cos (m \lambda \pi ) \right. \nonumber \\
	& \left. \quad - \text{Si}(m \lambda \pi) \sin (m \lambda \pi) \  
	 + C  + \log (\lambda \pi)     \right \}, \label{closed_form_R}
	\end{align}
	where $(a)$ is obtained from \cite[eq. 4.331, eq. 4.338]{ToI}. Ci($x$) and Si($x$) are cosine integral and sine integral, respectively, and C is an Euler constant. Using \eqref{closed_form_R}, the transmission latency $Y$ is given by
	\begin{align}
	Y \hspace{-0.0mm} &= \frac{D \log 2}{W} \hspace{-0.0mm} \left\{ \log m \hspace{-0.0mm} - \hspace{-0.0mm} \text{Ci}(m \lambda \pi) \cos (m \lambda \pi ) \hspace{-0.0mm} \right. \nonumber \\
	& \left. \quad -  \hspace{-0.0mm} \text{Si}(m \lambda \pi) \sin (m \lambda \pi) \hspace{-0.0mm}  + \hspace{-0.0mm} C \hspace{-0.1mm} + \hspace{-0.0mm} \log (\lambda \pi)\right\}^{-1}.
	\end{align}
	Note that $Y$ changes according to the target \ac{STP} $\zeta$.

		\begin{figure}
		\vspace{-0.0cm}
		\centering
		\begin{center}
			\includegraphics[width=0.70\columnwidth]{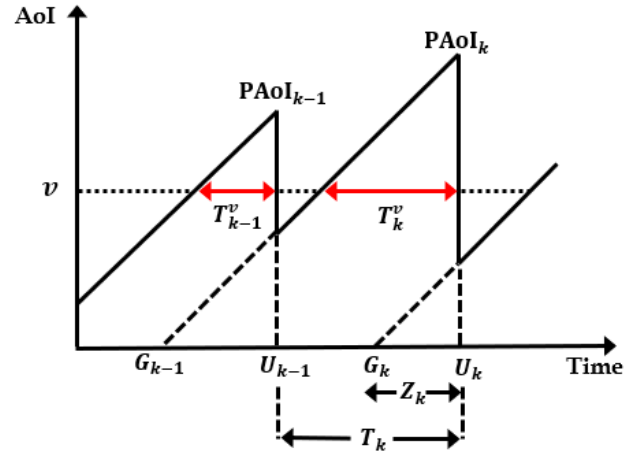}
		\end{center}
		
		\vspace{-0.4cm}
		\caption{A sample path of the \ac{AoI} where $G_k$ and $U_k$ are the generation instant and update instant of the $k\text{th}$ effective packet.  }
		
		\label{fig:sample path}
		\vspace{-0.5cm}
	\end{figure}
	\begin{figure}
		\centering
		\begin{center}
			\includegraphics[width=0.80\columnwidth]{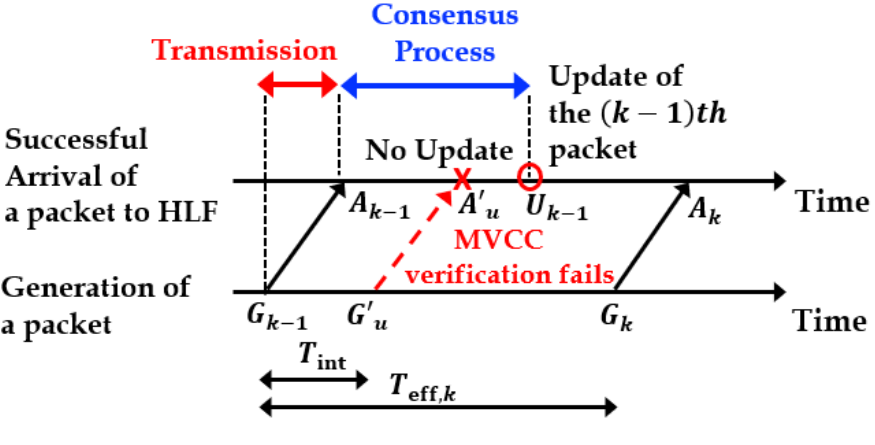}
		\end{center}
		
		\caption{The packet generated at $G'_u$ is become invalid due to
			the \ac{MVCC} verification.}
		
		\label{fig:MVCC}	
		\vspace{-0.5cm}
	\end{figure}

	\section{AoI Analysis of HeMN}
	 In this section, we formulate the \ac{AoI} violation probability for \ac{HeMN}.
	 As a metric for measuring the data freshness, the \ac{AoI} is utilized.
	 We focus on a specific status information, which is updated with a certain key value. We assume the source generates a packet with exponentially distributed inter-generation time with rate $\rho_\text{s}$. The inter-generation time of the two consecutive packets that successfully arrive at the \ac{BS} from the source is denoted as $T_{\text{int}}$ and we consider exponentially distributed $T_\text{int}$ with rate $\rho = \rho_\text{s} p_c$. We assume the consensus latency $X_k$ and $T_\text{int}$ are independent.
	 
	 As depicted in Fig. \ref{fig:MVCC}, not every generated packet can make a valid update because of the \ac{MVCC} verification.
	 We call the packets that make valid updates as \textit{effective packets}. For the $k\text{th}$ effective packet, we denote $G_k$ as the generation instant at the source, $A_k$ as the arrival instant at the \ac{BS}, and $U_k$ as the update instant for the ledger. In addition, $G'_u$ and $A'_u$ are for the generation and arrival instant of the $u\text{th}$ invalid packet, respectively. Then, the inter-generation time of effective packets $T_{\text{eff}, k}$ can be defined as $T_{\text{eff}, k} = G_k - G_{k-1}$. The $k\text{th}$ packet can be effective only if its arrival instant $A_{k}$ is after  $U_{k-1}$, which is the update instant of the previous $(k-1)\text{th}$ effective packet. Considering the transmission latency $Y,$ any generated packet after $Z_{k-1}-Y$ can be an effective packet. 
	 Due to the memoryless property of $T_{\text{int}}$,
	 $T_{\text{eff}, k}$ can also be defined as
	 \begin{align}
	 T_{\text{eff}, k} = Z_{k-1} - Y + T_{\text{int}},  
	 \end{align}

     The \ac{PAoI} is the \ac{AoI} before the update instant \cite{MAHD:19}. For the $k\text{th}$ effective packet, \ac{PAoI}$_{k}$ is  expressed as $ T_{\text{eff}, k} + Z_k $ as shown in Fig. \ref{fig:sample path}. To see whether a certain level of freshness is guaranteed in \ac{HeMN},
     we can evaluate the \ac{AoI} violation probability in \ac{HeMN} based on the sample path analysis introduced in \cite{JaHuJa:19}. Then, the \ac{AoI} violation probability is given by  \cite{JaHuJa:19}
	\begin{align}
	\P[\text{AoI} \geq v] = \frac{\E[T^v_k]}  { \E[T_{k}] }, \label{ratio_of_E}
	\end{align}
	where $T^v_k$ is the time duration of the \ac{AoI} being larger than the target \ac{AoI} $v$ between the update instants $U_{k-1}$ and $U_k$, and $T_k$ is the interval between update instants $U_{k-1}$ and $U_k$. In \ac{HeMN}, $T^v_k$ can be defined as 
	\begin{align}
	T^v_k &= \min{ \left\{ \left(\ac{PAoI}_k  -   v \right)^+,   T_k  \right\} }  \nonumber   \\
	&= \min{ \left\{ \left( X_{k-1} \hspace{-0.5mm} + \hspace{-0.5mm} X_{k} \hspace{-0.5mm} +  \hspace{-0.5mm}T_\text{int}  \hspace{-0.5mm} + \hspace{-0.5mm} Y  \hspace{-0.5mm}-  \hspace{-0.5mm} v \right)^+,   X_{k}  \hspace{-0.5mm}+ \hspace{-0.5mm} T_\text{int}         \right\}          }, \label{min_T^v_k}
	\end{align}
	where $(\cdot)^+ = \max {(0, \ \cdot) }$. As shown in Fig. \ref{fig:sample path}, $T_k$ can be given by
	\begin{equation}
	T_k = T_{\text{eff},k} + Z_{k} - Z_{k-1} =X_k + T_\text{int}.
	\end{equation}
	We now obtain the \ac{AoI} violation probability of \ac{HeMN} in the following theorem.
	\begin{theorem}
		In \ac{HeMN}, the \ac{AoI} violation probability is given by \eqref{AoI_violation_Prob} on the top of this page.

		\begin{figure*}[!t]
		\begin{align}
		\P[\text{AoI} \geq v] &= \frac{\rho \beta^{2\alpha+1}}{\left( \beta \hspace{-0.5mm} + \hspace{-0.5mm} \rho \alpha \right) \Gamma(\alpha)^2} \summation{n=0}{\infty} \frac{\left( \rho \hspace{-0.5mm}  - \hspace{-0.5mm} \beta  \right)^n}{n! (\alpha \hspace{-0.5mm} + \hspace{-0.5mm} n )   \rho^{\alpha +n+1} } \left[  \frac{\Gamma(\alpha \hspace{-0.5mm} + \hspace{-0.5mm} n \hspace{-0.5mm} + \hspace{-0.5mm} 1)}{\beta^\alpha} \gamma\left(\alpha, \beta T_c   \right) - \summation{k=0}{\infty} \frac{(-1)^k (\rho T_c)^{\alpha +  n   +  k   +  1}}{k! \left( \alpha  \hspace{-0.5mm} + \hspace{-0.5mm} n \hspace{-0.5mm} + \hspace{-0.5mm}  k \hspace{-0.5mm} +  \hspace{-0.5mm} 1 \right) } B\left( \alpha  \hspace{-0.5mm} + \hspace{-0.5mm} n \hspace{-0.5mm} +  \hspace{-0.5mm}  k \hspace{-0.5mm}  + \hspace{-0.5mm} 2 , \alpha  \right) \right. \nonumber \\ 
		&\left. \times  T_c^{\alpha}   {_1}F_1 \left(\alpha; 2\alpha  \hspace{-0.5mm} + \hspace{-0.5mm} n \hspace{-0.5mm} + \hspace{-0.5mm} k  \hspace{-0.5mm} + \hspace{-0.5mm}  2; -\beta T_c \right)  \right] \hspace{-0.5mm} + \hspace{-0.5mm}  \frac{\rho}{(\beta \hspace{-0.5mm}  + \hspace{-0.5mm} \rho \alpha) \Gamma(\alpha)}  \left\{ \alpha \gamma(\alpha, \beta T_c)  \hspace{-0.5mm} -  \hspace{-0.5mm}   \summation{n=0}{\infty} \frac{ (-\beta T_c)^{2\alpha+n+1}} {n! \Gamma(\alpha) (\alpha  \hspace{-0.5mm} + \hspace{-0.5mm} n \hspace{-0.5mm}  + \hspace{-0.5mm} 1)} B(\alpha \hspace{-0.5mm} + \hspace{-0.5mm} n \hspace{-0.5mm} + \hspace{-0.5mm} 2, \alpha) \right. \nonumber \\
		& \left. \times {_1}F_1 (\alpha; 2\alpha \hspace{-0.5mm} + \hspace{-0.5mm} n \hspace{-0.5mm} + \hspace{-0.5mm} 2; -\beta T_c) - (\beta T_c)^{\alpha+1}  B(\alpha, 2) T_c^{\alpha} {_1}F_1 (\alpha; \alpha  \hspace{-0.5mm} + \hspace{-0.5mm} 2; -\beta T_c) \right. \nonumber \\
		&\left. +  \summation{n=0}{\infty} \frac{ (-\beta T_c)^{2\alpha+n+1}}{  n! \Gamma(\alpha) (\alpha+n)} B(\alpha , \alpha \hspace{-0.5mm}  +\hspace{-0.5mm} n \hspace{-0.5mm} + \hspace{-0.5mm} 2) {_1}F_1(\alpha; 2\alpha \hspace{-0.5mm} + \hspace{-0.5mm} n \hspace{-0.5mm} + \hspace{-0.5mm} 2; \beta T_c) \right\} + \frac{\Gamma(\alpha, \beta T_c)}{\Gamma(\alpha)}. \label{AoI_violation_Prob}
		\end{align}

		\centering \rule[0pt]{18cm}{0.3pt}
		
	\end{figure*}

	\end{theorem}
	\begin{proof}
		In \eqref{ratio_of_E}, the expectation of $T^v_{k}$ is given by
		\begin{align}
		\E[T^v_{k}] = \integral{0}{\infty} \P[T^v_{k} \geq a] da.  \label{expectation_t_k^v}
		%
		%
		\end{align}
		%
		%
		%
		%
		Using \eqref{min_T^v_k} and \eqref{expectation_t_k^v}, $\E[T^v_{k}]$ can be represented by
		%
		%
		%
		\begin{align}
		\E[T^v_{k}]  &= \integral{0}{T_\text{c}} \hspace{-2.0mm} \integral{0}{\infty} \hspace{-1.5mm} \P \left[ \left( x \hspace{-0.5mm} + \hspace{-0.5mm} X_{k} \hspace{-0.5mm} + \hspace{-0.5mm} T_\text{int} \hspace{-0.5mm} - \hspace{-0.5mm}  T_\text{c} \right) \geq a \right] f_{X_{k-1}}(x) \ da dx \nonumber \\
		& \quad + \integral{T_\text{c}}{\infty} \hspace{-1.0mm} \integral{0}{\infty} \P[X_k + T_{int} \geq a] f_{X_{k-1}} (x)\ da dx,  \label{T^v_k} 
		\end{align}
		where $f_{X_{k-1}}(x)$ is in \eqref{gamma_pdf}, and $T_\text{c} =v-Y$.
		%
		%
		%
		In \eqref{T^v_k}, $\P \left[ \left( x \hspace{-0.5mm} + \hspace{-0.5mm} X_{k} \hspace{-0.5mm} + \hspace{-0.5mm} T_\text{int} \hspace{-0.5mm} - \hspace{-0.5mm}  T_\text{c} \right) \geq a \right]$ can be given by
		%
		\begin{align}
		&\P \left[ T_\text{int}  \geq a + T_\text{c} - x - X_k \right]  \nonumber \\
		& \hspace{-0.9mm} \overset{(a)}{=} \hspace{-0.5mm}  \integral{0}{a + T_\text{c} - x} \hspace{-2.0mm}  e^{-\rho \left( a + T_\text{c} -x -w  \right)}  f_{X_k} (w) dw + \integral{a + T_\text{c} - x}{\infty} \hspace{-2.0mm}  f_{X_k} (w) dw \nonumber \\ 
		& \hspace{-0.9mm} \overset{(b)}{=} \hspace{-0.9mm}  \frac{\beta^\alpha \hspace{-0.5mm} e^{-\rho \left(a \hspace{-0.3mm} +  \hspace{-0.3mm} T_\text{c} \hspace{-0.3mm} - \hspace{-0.3mm} x  \right)  }}{\Gamma(\alpha)} \frac{\gamma \hspace{-0.5mm} 
		\left(
		\hspace{-0.5mm} \alpha, \left( \hspace{-0.5mm} \beta \hspace{-0.5mm} - \hspace{-0.5mm} \rho   \hspace{-0.5mm}     \right) \hspace{-0.5mm}  \left( \hspace{-0.5mm} a \hspace{-0.5mm}  + \hspace{-0.5mm}  T_\text{c}  \hspace{-0.5mm} -   \hspace{-0.5mm}  x  \hspace{-0.5mm} \right)   \hspace{-0.5mm}   
		\right)
	    }{(\beta - \rho)^\alpha} \hspace{-0.5mm} + \hspace{-0.5mm} \frac{\Gamma \hspace{-0.7mm} 
		\left(
		\hspace{-0.5mm} \alpha, \beta \hspace{-0.5mm} \left( \hspace{-0.5mm}  a \hspace{-0.5mm}  + \hspace{-0.5mm}  T_\text{c} \hspace{-0.5mm}  - \hspace{-0.5mm}  x  \hspace{-0.5mm} \right) \hspace{-0.5mm}  
		\right)
	    }{\Gamma(\alpha)}, \label{Just Prob}
		\end{align}
		where $(a)$ is obtained from the fact that $\P \left[ T_\text{int} \geq a  + T_c - x - X_k\right]$ is always one when $X_k$ is larger than $a + T_\text{c} - x$, and the exponential distribution of $T_\text{int}$. In \eqref{Just Prob}, $\gamma(\cdot, \cdot)$ and $\Gamma(\cdot, \cdot)$ are the incomplete gamma functions, given by
		%
		\begin{align}
		\gamma(\alpha, x) \hspace{-0.5mm} =  \hspace{-0.5mm} \integral{0}{x}  \hspace{-0.5mm} t^{\alpha - 1}  \hspace{-0.5mm} e^{-t} dt , \ \Gamma(\alpha, x)  \hspace{-0.5mm} =  \hspace{-0.5mm} \integral{x}{\infty}  \hspace{-0.5mm} t^{\alpha - 1}  \hspace{-0.5mm} e^{-t} dt.
		\end{align}
		
		Using \eqref{Just Prob}, the inner integral of the first term in \eqref{T^v_k} can be obtained as
		%
		%
		\begin{align}
		&\integral{0}{\infty} \hspace{-0.5mm} \P \left[ \left( x \hspace{-0.5mm} + \hspace{-0.5mm} X_{k} \hspace{-0.5mm} + \hspace{-0.5mm} T_\text{int} \hspace{-0.5mm} - \hspace{-0.5mm}  T_\text{c} \right) \geq a \right] da \nonumber \\ 
		%
		 &= \integral{0}{\infty} \left[ \frac{\beta^\alpha e^{-\rho \left(  a + T_\text{c} - x  \right)  }}{\Gamma(\alpha)} \frac{\gamma 
		 \left(
		 \alpha, \left(  \beta  -\rho        \right)  \left( a + T_\text{c} -x   \right)      
		 \right)
	     }{(\beta - \rho)^\alpha} \right. \nonumber \\ 
		&\left. \quad + \frac{\Gamma
		\left(
		\alpha, \beta \left( a + T_\text{c} - x  \right)
		\right)
	    }{\Gamma(\alpha)} \right] da \nonumber \\
		&\overset{(a)}{=} \frac{\beta^\alpha }{ \Gamma(\alpha)} \summation{n=0}{\infty} \frac{ \left(   \rho \hspace{-0.5mm}  -  \hspace{-0.5mm}  \beta \right)^{n} }   {n! \left(  \alpha  \hspace{-0.5mm}  + \hspace{-0.5mm}  n  \right)} \integral{0}{\infty} \hspace{-2.5mm}   \left(a \hspace{-0.5mm}  +  \hspace{-0.5mm} T_\text{c} \hspace{-0.5mm}   -  \hspace{-0.5mm}  x  \right)^{\alpha + n} e^{-\rho \left( a + T_\text{c} -  x  \right)} da \nonumber \\
		& \quad + \frac{1}{\Gamma(\alpha)} \integral{0}{\infty}  \integral{\beta(a + T_\text{c} -x )}{\infty}   t^{\alpha-1} e^{-t} dt da   \nonumber  \\
		&= \underbrace{\frac{\beta^\alpha}{\Gamma(\alpha) } \summation{n=0}{\infty} \frac{\left(  \rho - \beta \right)^n}{n! \left( \alpha + n \right)} \frac{\Gamma
		\left(
		\alpha + n + 1, \rho \left(T_\text{c} - x \right)   
		\right)
	    }{\rho^{\alpha +n +1}}}_{f_1(x)} \nonumber \\ 
		& \quad + \underbrace{\frac{ \Gamma 
		\left(
		\alpha \hspace{-0.5mm} + \hspace{-0.5mm} 1, \beta \left( T_\text{c} \hspace{-0.5mm} - \hspace{-0.5mm} x \right)    
		\right)
	    }  {\beta \Gamma(\alpha)}}_{f_2(x)} + \underbrace{\frac{ \left( x \hspace{-0.5mm} - \hspace{-0.5mm} T_\text{c} \right) }{\Gamma(\alpha) } \Gamma
        \left(
        \alpha, \beta \left( T_\text{c} \hspace{-0.5mm} - \hspace{-0.5mm} x  \right)   
        \right)
        }_{f_3(x)},  \label{substitution}
		\end{align}
		where $(a)$ is from the fact that $\gamma(\alpha, x)$ can be represented as \cite[eq. 8.354-1]{ToI}
		%
		%
		\begin{align}
		\gamma(\alpha, x) = \summation{n=0}{\infty} \frac{(-1)^n x^{\alpha + n}}{n! (\alpha + n)}. \label{8.354-1}
		\end{align}
		Using \eqref{substitution}, the first term in \eqref{T^v_k} can be represented as
		\begin{align}
		\integral{0}{T_\text{c}} \left\{f_1(x) + f_2(x) +  f_3(x) \right\} f_{X_{k-1}}(x) dx , \label{f_substitution}
		\end{align}
		where $f_i(x) , i=1, 2, 3$ is represented in \eqref{substitution}. In \eqref{f_substitution}, the first term can be given by 
		\begin{align}
		&\integral{0}{T_\text{c}} \hspace{-1.5mm}  f_1(x) f_{X_{k-1}}(x) dx \nonumber \\
		%
		%
		&\overset{(a)} {=} \frac{\beta^{2\alpha}}{\Gamma(\alpha)^2 } \summation{n=0}{\infty} \frac{\left(  \rho - \beta \right)^n}{n! \left( \alpha + n \right) \rho^{\alpha +n +1}} \integral{0}{T_\text{c}} \bigg[  \Gamma\left( \alpha +n + 1 \right)   \nonumber 
		\end{align}

		\begin{align}
		& \quad  - \summation{k=0}{\infty} \frac{(-1)^k \left\{ \rho \left ( T_\text{c} -x \right)  \right\}^{\alpha+n+k+1} }{k! \left( \alpha +n+k+1 \right)}  \bigg ] x^{\alpha - 1} e^{-\beta x} dx  \nonumber \\
		%
		& \overset{(b)} {=} \frac{\beta^{2\alpha}}{\Gamma(\alpha)^2 } \summation{n=0}{\infty} \frac{\left(  \rho - \beta \right)^n}{n! \left( \alpha + n \right) \rho^{\alpha +n +1}} \bigg[ \frac{ \Gamma\left( \alpha \hspace{-0.5mm} + \hspace{-0.5mm} n \hspace{-0.5mm}  + \hspace{-0.5mm}  1 \right)} {\beta^\alpha} \gamma\left(\alpha, \beta T_\text{c} \right)  \nonumber \\
		& \quad -\summation{k=0}{\infty} \frac{(-1)^k (\rho T_\text{c})^{\alpha+n+k+1} }{k! \left( \alpha +n+k+1 \right)} T_\text{c}^{\alpha} B\left(\alpha \hspace{-0.5mm} + \hspace{-0.5mm} n \hspace{-0.5mm} + \hspace{-0.5mm} k \hspace{-0.5mm} + \hspace{-0.5mm} 2, \alpha \right)   \nonumber \\
		& \quad \quad \times {_1}F_1 \left( \alpha; 2\alpha  \hspace{-0.5mm} + \hspace{-0.5mm}  n  \hspace{-0.5mm} +  \hspace{-0.5mm} k  \hspace{-0.5mm}  + \hspace{-0.5mm}  2; \beta T_\text{c} \right) \bigg],  \label{1_3}
		\end{align}
		%
		where $B(\cdot, \cdot)$ is the beta function and ${_1}F_1 (\cdot \ ; \cdot \ ; \cdot)$ is the confluent hypergeometric function. The derivation of $(a)$ is obtained since $\Gamma(\alpha, x) = \Gamma(\alpha) - \gamma(\alpha, x)$ and $(b)$ is from \cite[eq. 3.383-1]{ToI}. Similarly,	
		in \eqref{f_substitution}, the second term can be represented by
		\begin{align}
		&\integral{0}{T_\text{c}} f_2(x) f_{X_{k-1}}(x)dx \nonumber \\
		&=\frac{\beta^{2\alpha}}{\Gamma(\alpha)^2} \hspace{-1.5mm} \integral{0}{T_\text{c}} \hspace{-1.2mm} \left[ \hspace{-0.5mm} \Gamma \hspace{-0.5mm} (\hspace{-0.5mm} \alpha \hspace{-0.5mm}+ \hspace{-0.5mm}  1 \hspace{-0.5mm} ) \hspace{-0.7mm} - \hspace{-0.7mm} \summation{n=0}{\infty} \hspace{-0.7mm} \frac{(-\beta)^n \hspace{-0.5mm} \left ( T_\text{c} \hspace{-0.5mm}  -  \hspace{-0.5mm} x \right)^{\alpha \hspace{-0.5mm} + \hspace{-0.5mm} n \hspace{-0.5mm} + \hspace{-0.5mm} 1} }{n! \left( \alpha  \hspace{-0.5mm} +  \hspace{-0.5mm}n \hspace{-0.5mm} + \hspace{-0.5mm} 1 \right)} \right] \hspace{-0.7mm} x^{\alpha \hspace{-0.5mm} - \hspace{-0.5mm} 1}  \hspace{-0.5mm} e^{-\beta x} \hspace{-0.5mm} dx \nonumber \\
		&= \frac{\alpha \gamma(\alpha, \beta T_\text{c})}{\beta \Gamma(\alpha)} - \frac{T_\text{c}}{\Gamma(\alpha)^2} \summation{n=0}{\infty} \frac{(-1)^n (\beta T_\text{c})^{2\alpha+n} }{n! \left( \alpha +n+1 \right)} \nonumber \\
		&\quad \times B(\alpha \hspace{-0.5mm} +  \hspace{-0.5mm}  n \hspace{-0.5mm} +  \hspace{-0.5mm} 2, \alpha) {_1}F_1(\alpha; 2\alpha \hspace{-0.5mm} +  \hspace{-0.5mm} n  \hspace{-0.5mm}+ \hspace{-0.5mm} 2 ; -\beta T_\text{c}). \label{2_2}
		\end{align}   
     	The third term in \eqref{f_substitution} is given by
		\begin{align}
		&\integral{0}{T_\text{c}} f_3(x) f_{X_{k-1}}(x) dx \nonumber \\
		%
		&\overset{(a)}{=} -\frac{\beta^\alpha}{\Gamma(\alpha)^2} \integral{0}{T_\text{c}} k \Gamma(\alpha, \beta k) (T_\text{c} -k)^{\alpha - 1} e^{-\beta (T_\text{c} - k)} dk \nonumber \\
		&\overset{(b)}{=} -\frac{\beta^\alpha}{\Gamma(\alpha)^2} \integral{0}{T_\text{c}} \left\{ \Gamma(\alpha)  - \summation{n=0}{\infty} \frac{(-1)^n (\beta k)^{\alpha +n}  }{n! (\alpha +n)}  \right\} \nonumber \\
		& \quad \quad \times 
		k (T_\text{c} - k)^{\alpha - 1} e^{-\beta (T_\text{c} - k)} dk \nonumber\\
		&\overset{(c)}{=} -\frac{T_\text{c}(\beta T_\text{c})^\alpha }{\Gamma(\alpha)} B(\alpha, 2)  {_1}F_1(\alpha; \alpha \hspace{-0.5mm} + \hspace{-0.5mm} 2; -\beta T_\text{c}) \nonumber \\
		& \quad + \frac{T_\text{c}}{\Gamma(\alpha)^2} \summation{n=0}{\infty} \frac{(-1)^n (\beta T_\text{c})^{2\alpha +n}  }{n! (\alpha +n)} B(\alpha, \alpha \hspace{-0.5mm} +  \hspace{-0.5mm} n \hspace{-0.5mm}  + \hspace{-0.5mm} 2) \nonumber \\
		&\quad  \quad \times {_1}F_1(\alpha; 2\alpha  \hspace{-0.5mm} + \hspace{-0.5mm} n \hspace{-0.5mm} + \hspace{-0.5mm} 2; -\beta T_\text{c}), \label{3_3}
		\end{align}
	where $(a)$ is obtained by substituting $k$ for $T_\text{c} - x$, $(b)$ is from \eqref{8.354-1} and $(c)$ is obtained by the similar steps used in \eqref{1_3} and then simplified by using ${_1}F_1 (a;b;z) = e^z {_1}F_1(b-a;b;-z)$ in \cite[9.212-1]{ToI}.
	%
	%
	%

		Now, we derive the second term of \eqref{T^v_k}. In the integral range of $a$, $X_k + T_\text{int} \geq 0$ holds. Hence, we have
		\begin{align}
		&\integral{T_\text{c}}{\infty} \integral{0}{\infty} \P[X_k + T_{int} \geq a] da f_{X_{k-1}} (x) dx \nonumber  \\ 
		&\overset{(a)}{=}  \left( \E[X_k] + \E [T_\text{int}] \right) \integral {T_\text{c}}{\infty} f_{X_{k-1}} (x) dx \nonumber \\ 
		&\overset{(b)}{=} \left( \frac{\alpha}{\beta} + \frac{1}{\rho}  \right) \frac{\Gamma(\alpha, \beta T_\text{c})}{\Gamma(\alpha)} \label{T_ak_2},
		\end{align}
		where $(a)$ is obtained since $X_k$ and $T_\text{int}$ are independent and $(b)$ is from the fact that $X_k$ and $T_\text{int}$ are the Gamma and exponential random variable, respectively.
		%
		
		Finally, $\E[T^v_{k}]$ is the summation of \eqref{1_3}, \eqref{2_2}, \eqref{3_3}, and \eqref{T_ak_2}.
		Note that the $\E[T_k] = \E[X_k + T_\text{int}] = \frac{\alpha}{\beta} + \frac{1}{\rho}$. Therefore, we obtain \eqref{AoI_violation_Prob} by the ratio of $\E[T^v_{k}]$ and $\E[T_k]$ as \eqref{ratio_of_E}
	\end{proof}
	Note that the \ac{AoI} violation probability in Theorem 1 can be applied to general \ac{HLF} networks with version 1.0 or higher \footnote{\ac{HLF} with version 1.0 or higher includes the MVCC verification}.
\section{Numerical Results}
In this section, we verify the \ac{AoI} violation probability analysis and show the impact of $\zeta$ on it.
We set $\rho_s = 15$, $P = 1$ W, $N_0 = -100$ dBm, $W = 1$ MHz, $D = 500$ Kb, and $\lambda = 0.0001$ (BS/km$^2$) as default values. We implement experiments for data freshness analysis in a \ac{HLF} platform v1.3 on one physical machine with Intel(R) Xeon W-2155 @ 3.30GHz. 
We generate the transactions with exponentially distributed inter-generation time to update a certain target key-value, which takes 30 percent of the whole generation. We also measure the consensus latency of the generated target transactions for different $\zeta$ from 0.3 to 1. Then, we fit statistical distribution for the thousand samples at different $\zeta$. The corresponding estimated parameters are averaged over five runs and shown in Table \ref{table:2}

  \begin{table}
 	\vspace{0.1cm}
  	\caption{Estimated shape ($\alpha$) and rate parameters ($\beta$) for a Gamma distribution } \label{table:2}
	\vspace{-0.2cm}
	\begin{center}

		\renewcommand{\arraystretch}{1.3}
	\begin{tabular}{|c|c|c|c|c|}
	\cline{1-2} \cline{4-5}
	\begin{tabular}[c]{@{}c@{}}Target\\ STP\end{tabular} & \begin{tabular}[c]{@{}c@{}}Average estimate\\ ($\alpha, \beta$)\end{tabular} &  & \begin{tabular}[c]{@{}c@{}}Target\\ STP\end{tabular} & \begin{tabular}[c]{@{}c@{}}Average estimate\\ ($\alpha, \beta$)\end{tabular} \\ \cline{1-2} \cline{4-5} 
	0.3                                                  & (5.64, 3.01)                                                      &  & 0.7                                                  & (7.18, 3.73)                                                      \\ \cline{1-2} \cline{4-5} 
	0.4                                                  & (5.94, 2.45)                                                      &  & 0.8                                                  & (7.71, 4.12)                                                      \\ \cline{1-2} \cline{4-5} 
	0.5                                                  & (5.39, 2.85)                                                      &  & 0.9                                                  & (7.50, 4.35)                                                      \\ \cline{1-2} \cline{4-5} 
	0.6                                                  & (5.42, 2.84)                                                      &  & 1.0                                                  & (6.57, 3.82)                                                      \\ \cline{1-2} \cline{4-5} 
	\end{tabular}
	\end{center}
	\vspace{-0.5cm}
\end{table}%

\begin{figure}
	\centering
	\vspace{-0.4cm}
	\begin{center}   
		{ 
			\psfrag{aaaaaaaaaaaaaaaaaaaaa}[Bl][Bl][0.59]   {Experiment, $\zeta$ = 0.4}
			\psfrag{a2}[Bl][Bl][0.59]   {Simulation, $\zeta$ = 0.4}
			\psfrag{a3}[Bl][Bl][0.59]   {Analysis, $\zeta$ = 0.4}
			\psfrag{b1}[Bl][Bl][0.59]   {Experiment, $\zeta$ = 0.6}
			\psfrag{b2}[Bl][Bl][0.59]   {Simulation, $\zeta$ = 0.6}
			\psfrag{b3}[Bl][Bl][0.59]   {Analysis, $\zeta$ = 0.6}
			\psfrag{c1}[Bl][Bl][0.59]   {Experiment, $\zeta$ = 0.8}
			\psfrag{c2}[Bl][Bl][0.59]   {Simulation, $\zeta$ = 0.8}
			\psfrag{c3}[Bl][Bl][0.59]   {Analysis, $\zeta$ = 0.8}
			\psfrag{Y1111111111}[bc][bc][0.7] {$\P[\text{AoI} \geq v]$}
			\psfrag{X1111111111}[tc][tc][0.7] {Target AoI, $v$}
			\includegraphics[width=0.93\columnwidth]{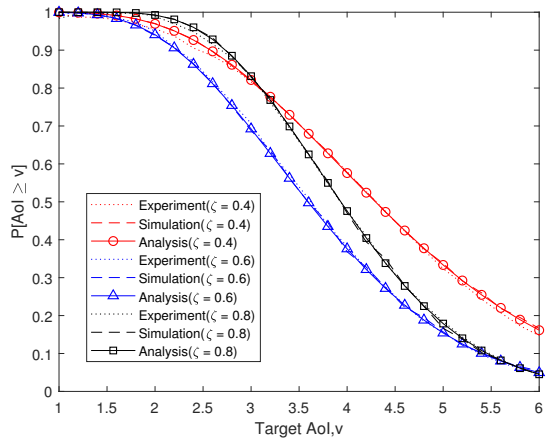}
		}
	\end{center}
	\vspace{-0.3cm}
	\caption{\ac{AoI} violation probability $\P[\text{\ac{AoI}} \geq v]$ as a function of a target \ac{AoI} $v$ for different target \ac{STP}s $\zeta$.}
	
	\label{fig: STP_validation}	
	\vspace{-0.5cm}
\end{figure}

In Figure \ref{fig: STP_validation}, we compare the analysis results with the simulation and experiment results as a function of $v$ at different $\zeta$. In the simulations, the sample path of the \ac{AoI} for the target key-value is generated and iterated over 10000 times in MATLAB environment. Fig. \ref{fig: STP_validation} shows that the analysis results match the experiment results obtained from the \ac{HLF} platform. Hence, the assumption that the consensus latency follows a Gamma distribution is reasonable. Each \ac{AoI} violation probability shows the different decay rates because the estimated parameters and $\zeta$ are different. It can be seen that the \ac{AoI} violation probability at $\zeta = 0.8$ is larger than the one at $\zeta = 0.4$ for low $v$. However, the result at $\zeta = 0.8$ shows the lowest \ac{AoI} violation probability for high target \ac{AoI}. Since the skewness of a Gamma distribution becomes reduced as the shape parameter increases, the mass of the distribution of consensus latency for $\zeta = 0.8$ tends to be concentrated on more right direction than others. Therefore, the result of $\zeta=0.8$ shows the steepest decay rate for increasing $v$. Data freshness in \ac{HeMN} can be more reliably guaranteed as $v$ increases for higher $\zeta$.

\begin{figure}[t]
    \vspace{-0.05cm}
	\centering
	\begin{center}
		\psfrag{X1111111111}[bc][bc][0.7] {Target STP $\zeta$}

		\psfrag{Y1111111111}[Bl][Bl][0.59]   {$\P[\text{AoI} \geq v]$}

		\includegraphics[width=0.93\columnwidth]{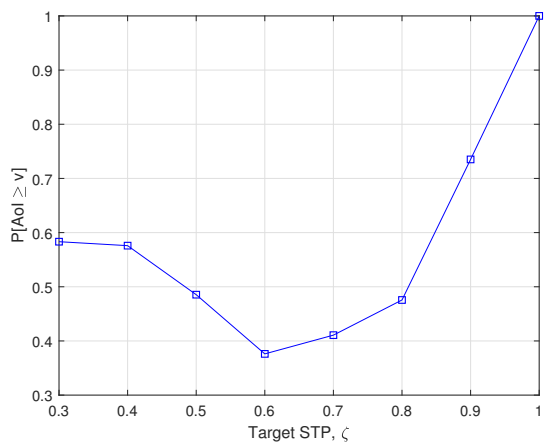}
	\end{center}
	\vspace{-0.3cm}
	\caption{\ac{AoI} violation probability versus target \ac{STP} $\zeta$ for the target \ac{AoI} $v = 4$.}
	
	\label{fig: STP}	
	\vspace{-0.6cm}
\end{figure}

Figure \ref{fig: STP} shows the \ac{AoI} violation probability for different values of the target \ac{STP}s with $v = 4$. We compute the \ac{AoI} violation probabilities at each $\zeta$ from the empirical samples. As can be seen in Fig. \ref{fig: STP},  $\zeta$  and the \ac{AoI} violation probability have a trade-off relation.
%
%
%
In specific, when $\zeta$ is small, the \ac{AoI} violation probability also becomes small due to the frequent outages of packets. As $\zeta$ increases, the \ac{AoI} violation probability decreases. Since $\zeta$ can guarantee a certain level of \ac{STP}, its increment results in more successfully received packets for the status update. However, after a certain point of $\zeta$, both $\zeta$ and the \ac{AoI} violation probability increase. This is from the fact that large $\zeta$ not only guarantees a reliable transmission but also increases transmission latency. For this reason, it becomes difficult for the consensus process to be completed before the target \ac{AoI} at high $\zeta$. Hence, the optimal target \ac{STP} $\zeta$ exists for minimizing the \ac{AoI} violation probability.
\vspace{-0.1cm}

\section{Conclusion}
In this paper, we have provided the closed-form expression of the \ac{AoI} violation probability in \ac{HeMN}. Based on the sample path analysis, we analyze the characteristics of the \ac{AoI} in \ac{HeMN} by considering both the transmission latency and consensus latency. We validate our analysis results through the simulations and the experiments after constructing the \ac{HLF} platform. The results show that our analysis results can provide the stochastic guarantee for data freshness in \ac{HeMN}. Moreover, we show that the target \ac{STP} and \ac{AoI} violation probability have a trade-off relation and high target \ac{STP} does not always guarantee data freshness. The \ac{AoI} violation probability analysis for \ac{HeMN} can be applied to \ac{AoI} sensitive applications such as temperature monitoring systems or traffic management for vehicle systems, where both fresh information and data integrity are necessary.

\vspace{-0.1cm}

\bibliographystyle{IEEEtran}
\bibliography{StringDefinitions,IEEEabrv,mybib}

\end{document}